\documentclass[11pt,fleqn]{article}

\usepackage{amsthm, amsmath, amssymb}
\usepackage{graphicx}
\usepackage{algorithm}
\usepackage{algorithmic}
\usepackage{xspace}
\usepackage{enumitem}
\usepackage{tabularx}
\usepackage{xcolor}
\usepackage{graphicx} 
\usepackage{tikz}
\usepackage{wrapfig}

\newtheorem{theorem}{Theorem}[section]
\newtheorem{lemma}[theorem]{Lemma}

\newtheorem{corollary}[theorem]{Corollary}
\newtheorem{definition}[theorem]{Definition}
\newtheorem{claim}[theorem]{Claim}

\newtheorem*{startheorem}{Theorem}

\begin{document}

\sloppy

\title{Dynamic Weighted Fairness with Minimal Disruptions}

\author{
Sungjin Im\thanks{Department of Computer Science and Engineering, University of California, Merced. \texttt{sim3@ucmerced.edu}}
\and
Benjamin Moseley\thanks{Tepper School of Business, Carnegie Mellon University. \texttt{moseleyb@andrew.cmu.edu}}
\and
Kamesh Munagala\thanks{Department of Computer Science, Duke University. \texttt{kamesh@cs.duke.edu}}
\and
Kirk Pruhs\thanks{University of Pittsburgh, Computer Science Department. \texttt{kirk@cs.pitt.edu}}
}

\date{}
\maketitle
\thispagestyle{empty}

\begin{abstract}
In this paper, we consider the following dynamic fair allocation problem: Given a sequence of job arrivals and departures, the goal is to maintain an approximately fair allocation of the resource against a target fair allocation policy, while minimizing the total number of {\em disruptions}, which is the number of times the allocation of any job is changed. We consider a rich class of fair allocation policies that significantly generalize those considered in previous work. 

We first consider the models where jobs only arrive, or jobs only depart. We present tight upper and lower bounds for the number of disruptions required to maintain a constant approximate fair allocation every time step. In particular, for the canonical case where jobs have weights and the resource allocation is proportional to the job's weight, we show that maintaining a constant approximate fair allocation requires $\Theta(\log^* n)$ disruptions per job, almost matching the bounds in prior work for the unit weight case. For the more general setting where the allocation policy only decreases the allocation to a job when new jobs arrive, we show that maintaining a constant approximate fair allocation requires $\Theta(\log n)$ disruptions per job. We then consider the model where jobs
can both arrive and depart. We first show strong lower bounds on
the number of disruptions required to maintain constant approximate
fairness for arbitrary instances.
In contrast we then show that there there is an algorithm 
that can maintain constant approximate fairness with 
 $O(1)$ expected disruptions per job if the weights of the jobs are
 independent of the jobs arrival and departure order. 
We finally show how our results can be extended to the setting with multiple resources. 
\end{abstract}

\section{Introduction} 

The formal study of fair resource allocation has advanced rapidly in recent years, motivated by applications to computer systems~\cite{DRF, beyondDRF, hierarchalScheduling, noComplaint, discreteJobs, heterogeneousDRF}. The basic theory of fair resource allocation has its roots in Economics~\cite{Varian,HyllandZ} and in scheduling results in computer science~\cite{Graham,minimumUtility}. However, modern applications such as data center scheduling have motivated considering new desiderata in fair resource allocation. 

In this paper, we consider a dynamic model for resource allocation, a topic that has received significant attention in recent literature~\cite{onlineCakeCutting, dynamicSocialChoice, noAgentLeftBehind, dynamicFairDivision, dynamicFairDivision2, EC18Psomas}.  In this model, which is again motivated by computing systems, each of $n$ jobs (or agents) may potentially arrive or depart from the system, so at every time step $t$ we are presented with a set of alive jobs $N^t \subseteq N$. The set $N^{t} \setminus N^{t-1}$ is the set of jobs that {\em arrive} at time $t$, and the set $N^{t-1} \setminus N^{t}$ is the set that {\em departs} at time $t$. We have a single divisible resource. 

There is some underlying fair share policy $I(j, t)$ which specifies the ideal fair share of the resource for job $j$ at time $t$.  At every time step $t$, an allocation policy/algorithm $A$  must determine $A(j, t)$, its allocation of the resource to a job $j \in N^t$. The policy $A$ must be online in that it can not rely on knowledge of the future. Ideally one would like $A$ to be perfectly fair, that is
it is always the case that $A(j, t) = I(j, t)$. However, a perfectly
fair allocation policy would generally lead to a
disruption, which is a change in the resource allocation of a job, of {\em every} job  when any job arrives or departs
(which is exactly when a job's fair share changes in most natural fair share policies). 
These disruptions can have significant overheads as they involve reassigning resources and changing the job states
~\cite{dynamicFairDivision,dynamicFairDivision2,quincy,borg,migration}. Due to the overhead, limiting the number of disruptions is a key design factor to most systems; for example, see~\cite{borg,schwarzkopf2013omega}. Therefore, we follow the lead of~\cite{noAgentLeftBehind, dynamicFairDivision, dynamicFairDivision2}, and investigate the minimum number of disruptions required to achieve approximate fairness.

\begin{definition}
For $c \ge 1$, an allocation policy $A$ is $c$-approximate if it always guarantees that $$A(j, t) \ge I(j, t)/c.$$
\end{definition}

\subsection{Background and Weighted Fairness}
Previous work~\cite{noAgentLeftBehind,dynamicFairDivision,dynamicFairDivision2} considered the case of {\em uniform fairness}, where $I(j,t) = \frac{1}{|N^t|}$. In particular, the work of~\cite{dynamicFairDivision} considered the question: Suppose $d$ disruptions are allowed per time step, what value of $c$ is achievable? They show that $c = (d+1) \ln\left(\frac{d+1}{d}\right)$. In particular, even when $d = 1$, a constant value of $c$ is achievable. Conceptually, this algorithm splits the allocation of the most allocated job in half when a new job arrives, and allocates the other half to the new job. The work of~\cite{dynamicFairDivision2} extends this to the case where $d < 1$.

\paragraph{Weighted Fairness.}
There are many situations where the appropriate notion of fairness is something other than  a uniform sharing of the resource(s). One natural/common example is weighted/proportional fairness. In this setting each job $j$ has weight $w_j$. 

Weights typically correspond to {\em priorities} that could be based on criteria such as willingness to pay for the resource, importance of the job, and so on. Furthermore, as we discuss below, weights also arise naturally in fair allocation contexts where there are {\em multiple} resources that could be complements or substitutes, and the utility (or rate) of a job is a function of the resources of each type allocated to the job.

In {\em weighted fair share policies}, a job's ideal fair share  is proportional to its weight, that is,
$$I(j, t) = \frac{w_j}{\sum_{k \in N^t} w_k}.$$
Uniform fairness is a special case of weighted fairness, where the weight of every job is 1.

\medskip
The weighted case presents new difficulties that are not encountered in the unweighted case. In the model where jobs only arrive, consider the arrival of a large weight job. This can cause the allocations of all jobs to change if we wish to approximate their fair share. Indeed, we need to relax the assumption that the number of disruptions per time step is small, to conditions that either bound the worst-case or the average number of disruptions {\em per job}.

Therefore, the natural questions we seek to answer are: 

\begin{itemize}
\item
What is the optimal bound on the number of disruptions per job for $O(1)$-approximate allocation policies with weighted fairness? 
\item
And even more generally, what is the optimal  bound on the number of disruptions per job for $O(1)$-approximate allocation policies with more general fair share policies? 
\end{itemize}

\subsection{Our Results} 
In this paper, we answer all the above questions by presenting tight results in increasingly complex models of fairness. Further, unlike the unweighted case, we need to distinguish between the settings where jobs only arrive from that where jobs are allowed to arrive and depart. Our main (and somewhat surprising) result is that in the model where jobs only arrive, it is indeed possible to achieve constant approximation to fairness with nearly constant number of disruptions per job. 
When jobs can both arrive and depart, we show that to achieve constant approximate
fairness an algorithm will have to disrupt a large number of jobs per arrival/departure
for some instances. 
In contrast we show that there there is an algorithm 
that can maintain constant approximate expected fairness with 
 $O(1)$ expected disruptions per job if the weights of the jobs are
 independent of the jobs arrival and departure order. 

\subsubsection{Weighted Fairness with Only Arrivals}
We first consider weighted fairness in the arrival-only model, where $N^{t-1} \subseteq N^t$ for all times $t$. The same results will apply to the symmetric departure-only model where $N^{t} \subseteq N^{t-1}$ for all times $t$.
(Imagine maintaining a fair allocation of some resource among a batch of jobs
as jobs finish and depart.) 
In section \ref{sect:logstar} we show the number of disruptions required to achieve approximate fairness  only increases by a very modest factor relative to uniform fair share. 

\begin{theorem}
\label{thm:arrival}
Consider weighted fair share policies in the arrival only model.
There is an $O(1)$-approximate allocation policy
that  will cause
at most 
 $O\left(\log^* n\right)$ disruptions for each job, where $n$ is the total number of arriving jobs. This result is tight, that is, every $O(1)$-approximate deterministic policy must suffer  $\Omega\left( \log^* n\right)$ disruptions per each job on average for some instance.
\end{theorem}

Our allocation policy groups jobs  into groups with exponentially increasing weights, and then treats each group as a single job. It then applies a monotone transform to the weight of each job, and uses this transformed weight instead of the original weight to perform the weighted fair allocation. The  transformation must both (a) be sufficiently invariant to  keep the number of disruptions low; and (b) sufficiently faithful to the original weight of the jobs to achieve $O(1)$-approximation. In fact, it is {\em a priori} not even clear that such a transform even exists, and showing its existence is one of our primary technical contributions.

\paragraph{Cobb-Douglas Utilities and Proportional Fairness.} Our allocation policy and its analysis easily extend to some canonical settings where there are $D$ divisible resources each with unit supply, and the rate of a job is a function of the resources allocated to it. 
One canonical rate model is {\em Cobb-Douglas}~\cite{sharingIncentives}, where job $j$ has a {\em substitutability} vector $\alpha_{jd}, d = [D]$ with $\sum_{d=1}^D \alpha_{jd} = 1$. Given allocation $x_{jd}$ in dimension $d$, the rate of execution is:
$$ y_j = \prod_{d=1} x_{jd}^{\alpha_{jd}}$$
A {\em proportionally fair} allocation~\cite{Im2014,sharingIncentives} maximizes $\prod_j y_j$.  It is easy to check that the resulting allocation has a closed form where:
$$ I_d(j,t) = x_{jd} = \frac{\alpha_{jd}}{\sum_{k \in N^t} \alpha_{kd}} \qquad \forall j, d \in [D]$$

Note now that this allocation independently performs a weighted fair allocation in each dimension $d$, where the weight of job $j$ in dimension $d$ is $\alpha_{jd}$. Further, it is easy to check that if the allocation is $c$-approximate in each dimension, then the resulting rate $y_j$ is also a $c$-approximation. Therefore, if we run our allocation policy independently in each dimension, the resulting policy is a constant approximation to the rate, and the resulting number of disruptions  is $O(nD \log^* n)$, where $n$ is the total number of arriving jobs.

\subsubsection{Weighted Fairness with Both Arrivals and Departures}
We next consider the case where jobs can both arrive and depart. In the uniform setting, the case with both arrivals and departures is not any harder than the arrival-only model. However, when we generalize to weighted fairness, this is no longer the case. In Section \ref{sect:arrivaldeparture} we prove
Theorem \ref{thm:departure}, which shows that with both arrivals and departures it is no longer possible
to always achieve both $O(1)$-approximation and a near linear number of disruptions. 

In contrast in Section \ref{sect:randomarrivaldeparture} we prove
Theorem \ref{thm:randomdeparture} that shows that this is possible if job weights
are independent of the jobs arrival and departure order. 

\begin{theorem}
\label{thm:departure}
Consider weighted fair share policies with both job arrivals and departures.
For every $c$-approximate deterministic algorithm $A$, there is an instance that causes $A$ to make $\Omega(  n^{1+ 1/(4c+1)})$ disruptions. 
\end{theorem}

\begin{theorem}
\label{thm:randomdeparture}
Consider weighted fair share policies with both job arrivals and departures.
Assume that the weights $w_1, \ldots, w_n$ of the jobs are arbitrary,
but the assignment of these weights to the $n$ jobs is uniformly random. 
In this setting there is an  $4$-approximate randomized algorithm $A$,
for which the expected number of disruptions per arrival and per departure is 
at most 5. 
\end{theorem}

\subsubsection{Monotone Fairness}
We next consider the number of disruptions needed to achieve approximate fairness for an arbitrary fair share function $I$ with arrivals only. 
The first thing to observe is that one can simulate departures by setting the fair share of a job to zero.
(Note that the lower bound in Theorem~\ref{thm:departure} extends to the case with both arrivals and departures.)  Thus to obtain some sort of positive result, one needs to impose some additional property on $I$. 
One natural  property that many/most fair share policies have
is monotonicity, that is, the arrival of a job can not increase another job's fair share, and the departure of a job can not decrease another job's fair share. More formally:

\begin{definition}
\label{def:monotone} 
A  fair resource share policy $I$ is monotone if it satisfies the following conditions:
Suppose job $j$ arrives at time $t$, then $I(j',t) \le I(j',t-1)$ for every $j' \in N^{t} \setminus \{j\}$. Similarly, if job $j$ departs at time $t$, then $I(j',t) \ge I(j',t-1)$ for every $j' \in N^{t} \setminus \{j\}$. 
\end{definition}

\begin{figure*}[htbp]
\centering
\begin{tabular}{|l|c|c|r|}
\hline 
Fairness Model & Dimensions & Arrival Model & Disruptions per Job \\
\hline 
Weighted Round Robin & $1$ & Arrival Only & $\Theta(\log^* n)$ \\
Monotone Fairness & $1$ & Arrival Only & $\Theta(\log n)$ \\
Cobb-Douglas + Proportional Fairness & $D$ & Arrival Only & $\Theta(D \log^* n)$ \\
Dominant Resource Fairness & $D$ & Arrival Only & $\Theta(D \log n)$ \\
\hline
Weighted Round Robin & $1$ & Arrival-Departure & $\Omega\left(n^{\frac{1}{4c+1}}\right)$ \\
Arbitrary Fair Policy & $1$ & Arrival Only & $\Omega\left(n^{\frac{1}{4c+1}}\right)$ \\
\hline
\end{tabular}
\caption{\label{tab:1} Summary of worst case number of disruptions per job needed to achieve $c$-approximate fairness for some constant $c > 1$. The lower bound on the penultimate line extends to monotone fairness. The final line should be interpreted as: There exists some fair share policy for which the number of disruptions is lower bounded by $\Omega\left(n^{\frac{1}{4c+1}}\right)$. Note that this table doesn't show Theorem~\ref{thm:randomdeparture} which states $O(1)$ disruptions per job on average when jobs are assigned random weights in the arrival-departure model.}
\end{figure*}

 In Section~\ref{sect:monotone} we show that while more disruptions may be needed to approximate fairness for an arbitrary monotone 
 fairness policy than for weighted fairness policies, it is still possible to achieve
 an almost linear number of disruptions.

\begin{theorem}
\label{thm:monotone}
Consider general monotone share policies in the arrival-only model.
There is a $O(1)$-approximate deterministic algorithm $A$ such that the number of disruptions per job is $O(\log  n)$. This bound is tight, that is, for every deterministic $O(1)$-approximate algorithm $A$, there are instances that cause $A$ to make $\Omega(\log n)$ disruptions per job on average.
\end{theorem}

\paragraph{Concave Utilities.}  We now give some examples of monotone fair share policies in a setting where the rate at which $j$ executes is a function $y_j = f_j(x_j)$ of the allocated resource amount $x_j$, where $f_j$ is non-decreasing and concave. This models the canonical cluster computing scenario where there are many identical machines, and parallelizable jobs~\cite{EdmondsP09,Im2014}. The rate of execution is a concave function of the amount of machines assigned to it.  Suppose the fair allocation algorithm either maximizes  $\min_{j \in N^t} y_j$, {\em i.e.}, is max-min fair, or maximizes the product of the rates, $\prod_{j \in N^t} y_j$, {\em i.e.}, is {\em proportionally fair}. Then it is easy to check that both these optima are achieved by water-filling on the $x_j$. Therefore, the resulting allocations $I(j,t) = x_j$ are monotone.

\paragraph{Dominant Resource Fairness.} In the case of multiple resources, we say that an allocation is monotone if it is monotone for each resource individually. One popular fair share policy is {\em weighted Dominant Resource Fairness} (weighted DRF)~\cite{DRF, beyondDRF} that generalizes a max-min fair allocation. Suppose job $j$ has weight $w_j$ and resource requirement $r_{jd}$ in resource $d \in [D]$. Assume by scaling that there is one unit of resource available for each resource. If the job executes at rate $y_j$, it consumes an amount $r_{jd} y_j$ of resource $d$. The weighted DRF allocation sets $y_j$ so that:
\begin{enumerate}
\item $\sum_j r_{jd} y_j \le 1$ for all dimensions $d$;
\item $w_j y_j \max_d r_{jd}$ is the same for all jobs, {\em i.e.}, the weighted share of the dominant resource consumed is equalized.
\end{enumerate}
It is clear that these shares can be computed by water-filling on the $y_j$, so that the fair share $I_d(j,t) = r_{jd} y_j$ is monotone in each dimension $d$. 

Thus with $D$ resources, our results immediately imply a bound of $O(nD \log n)$ disruptions  needed to maintain a constant approximation to any monotone fair share policy, including the DRF policy in particular.

\subsection{Summary and Related Work}
We summarize our results in Table~\ref{tab:1}.  We have already discussed the work of~\cite{noAgentLeftBehind,dynamicFairDivision,dynamicFairDivision2}, which considers the unweighted case. 
The work of~\cite{Li:2018,Li:2018ArXiv} study a demand model that is superficially similar to weighted fairness.
In the demand model in \cite{Li:2018,Li:2018ArXiv} each job $j$ has a demand
$d_j$, representing the fraction of the resource that the job wants. 
An allocation is then $c$-fair if the fraction of the resource
that a job $j$ gets is at least $\min(d_j, d_j/(c \cdot d))$, where $d$ is
the total demand of the jobs in the system. They
show that $\Theta(\log n)$ disruptions are necessary and sufficient to 
maintain constant approximate fairness in this demand model.  However, their definition of jobs ``present" in the system at any point in time also includes jobs that departed in the past. In that sense, their model even with departures is comparable to our arrival only model, where our amortized bound of $\Theta(\log^* n)$ disruptions is an improved result. Otherwise, this demand model is not directly comparable to our work.

\section{Weighted Fairness: Arrival Model and Proof of Theorem~\ref{thm:arrival}}
\label{sect:logstar}

In this section we prove Theorem \ref{thm:arrival}, which we restate below for
convenience.

\begin{startheorem}
Consider weighted fair share policies in the arrival only model.
There is an $O(1)$-approximate allocation policy
that  will cause
at most 
 $O\left(\log^* n\right)$ disruptions for each job. This result is tight, that is, every $O(1)$-approximate deterministic policy must suffer  $\Omega\left(n \log^* n\right)$ total disruptions for some instance.
\end{startheorem}

\subsection{Upper Bound}

 In this subsection we give an algorithm that is $O(1)$-approximate and  ensures that the number of disruptions per job is $O( \log^* n)$. 
To build intuition, while postponing some messy details, we will first discuss some special cases. In particular, our presentation of the special case discussed in section \ref{sec:geo} is designed to explain the key algorithmic design and analysis insights  as simply  as possible. 

Throughout the paper, we assume that the number of jobs $n$ is known a priori. This assumption can be removed by the standard guess-and-double technique, where we use a guess of the number of jobs, say $1$ initially, and keep doubling the guess when the number of jobs exceed the previous value. It is easy to check that our analysis and bounds in the entire paper will hold with this modification to the algorithm, and we omit the details.

\subsubsection{Geometrically Increasing Weights}
\label{sec:geo}
First consider the  special case where the weight of job $j$ that arrives at time $j$ is $2^{j-1}$, which will eventually be our lower bound instance. Intuitively, the worst case instance should be a sequence of jobs whose respective weight keeps increasing considerably but not too drastically: If the increase is tiny, there's no need to disrupt the existing jobs as their fair shares change little when a new job arrives. Further, if the increase is huge, the existing jobs' total  fair share becomes negligible in the near future as opposed to the newly arriving jobs, meaning that disrupting jobs that arrived long ago doesn't help serve new jobs of huge weights.

We give an algorithm that is $O(1)$-approximate and  ensures that the total number of disruptions  is $O(n \log^* n)$.
Note that the total weight of the alive jobs at time $t$ is essentially $2^{t}$, and the ideal fair
share for job $j$ is essentially $I(j, t)= 1/2^{t-j+1}$. Note that the arrival of a new job decreases the fair share of existing jobs by a factor of $2$. A naive approach that maintains a constant factor approximation to these rates will attempt to always maintain a constant factor approximation to every rate, and would therefore reassign the rates of all existing jobs every constant number of steps. This means it will incur $\Omega(n^2)$ reassignments over $n$ jobs. 

\medskip
\noindent {\bf Intuition.} It is {\em a priori} not even obvious we can do any better. The key idea is now to construct a monotone map from $1/I(j,t)$ to a small set of integers, and use the inverse of this map as the rate. Since the set of integers is small and since rates are monotone in the arrival only model, the size of this set will bound the number of times we reallocate the rate. Of course, in this process, first, the new rate needs to be at least a constant factor of the original rate (they could be much larger, but cannot be much smaller); and  secondly, the resource should not be over-allocated given we are increasing rates.

Consider the function $g(x) = 2^{2^x}$. Then, for constant $d \ge 1$ consider modifying the rate $I(j,t) = 1/2^{t-j+1}$ to 
$$A(j,t) = \frac{1}{d} \frac{1}{g(\lfloor \log_2 \log_2 (1/I(j,t)) \rfloor)} = \frac{1}{d} \frac{1}{g(\lfloor \log_2 (t-j+1) \rfloor)}$$

First note that the value  $\lfloor \log_2 (t-j+1) \rfloor$ is an integer that is at most $O(\log n)$ since $t \le n$, so this bounds the number of reassignments per job. Second, without the floor, the expression above is exactly $\frac{I(j,t)}{d}$, and taking the floor only reduces the value of the denominator, so that $A(j,t) \ge \frac{1}{d} I(j,t)$. 

The tricky part is to bound the resource allocated. For this, the new rate cannot be too large for many jobs. In particular, we need to show that for all $1 \le t \le n$, we have
$$ \frac{1}{d} \sum_{j \le t} \frac{1}{g(\lfloor \log_2 (t-j+1) \rfloor)} \le 1$$
Expanding the above summation, we need to bound
\begin{eqnarray*}
 && \frac{1}{d} \left(\frac{1}{2^{2^0}} + \frac{2^2 - 2^1}{2^{2^1}} +  \frac{2^3 - 2^2}{2^{2^2}} + \frac{2^4 - 2^3}{2^{2^3}} + \cdots \right) \\
 & \le &  \frac{2}{d} \sum_{k \ge 0} \frac{2^k}{2^{2^k}} \le \frac{4}{d}
\end{eqnarray*}
Therefore, we have a $O(1)$ approximation to fairness while performing $O(\log_2 n)$ disruptions per job.

\medskip
\noindent {\bf Algorithm.} Our algorithm builds on the above intuition, and constructs a mapping with even smaller range of integers. The tradeoff is that this can lead to over-allocation of resource if we are not careful. Further, our algorithm has to work for any weights and not just those that are growing exponentially. The question is: How far can we push this idea?  Turns out, quite a lot!

Our final algorithm for this special case maintains the allocation 
$$A(j, t) = \frac{1}{d} \max \left(  \frac{1}{g( \lfloor g^{-1}( 1/I(j, t)) \rfloor)} , \frac{1}{n} \right)$$
Here $d$ is a constant, and $g^{-1}(x)$ is a  slowly growing function, whose final definition will be revealed by the analysis. 

As before, the analysis involves showing the following three facts. Note that this is just a sketch of analysis as we already gave a simpler and looser analysis and will give the analysis of our algorithm for the general case momentarily. 
\begin{itemize}
    \item This allocation is $\frac{1}{d}$-approximate. If the floor in the definition of the allocation was removed, then the allocation of a job would be the maximum of $\frac{1}{d}$ of the job's fair share and $\frac{1}{dn}$, which is obviously $\frac{1}{d}$-approximate.
    And the inclusion of the floor can not decrease the allocation.
    \item
    The resource is not over allocated.
    To show this it is sufficient to show $\sum_{j=1}^t \frac{1}{g( \lfloor g^{-1}( 1/I(j, t)) \rfloor)} \le d-1$. This is not completely straight forward, but one reasonable approach would be to bound the number of times that a term $1/g(k)$ can appear in this sum. To get some reasonable bound, $g^{-1}$ can not be too slowly growing.  After a bit of contemplation, one can see that it is sufficient to define $g(x)$ by: $\log_2 g(k+1) = g(k) / 2^k$. 
    Then, as in  the analysis above, the term  $1/g(k)$ can appear only $\log_2 g(k+1)$ times in the sum.\footnote{To see this, consider the terms in the summation in decreasing order of $j$. Note that the value of $I(j, t)$ decreases by a factor of 2 in this order.     Consider the two earliest terms of value $1 / g(k)$ and $1/ g(k+1)$. Then, if $n'$ is the number of appearances of $1 / g(k)$, we have $1 / (g(k) 2^{n'}) = 1 / g(k+1)$, which gives $n' \leq \log_2 g(k+1)$.}  Thus the summation is then bounded by $\sum_{k \geq 1} (1/g(k)) \cdot (g(k) / 2^k) \leq 2$.
    Thus it is sufficient to define $d=3$.
    \item
     No job is disrupted more than $\log^* n$ times.
This follows from noting three facts.
First,  that $g^{-1}(x) = \Theta(\log^* x)$.
Second,  if a job's allocation changes when its fair share is $1/s$, then its allocation will not change again until its fair share is something like $1/2^s$.
Third, the minimum allocation for  each job is $\frac{1}{dn}$. Then, if a job $j$'s initial fair share is $1 / s$, the number of times $j$'s allocation changes is maximized when it does at each time $j$'s fair share becomes $1 / 2^s$, $1 / 2^{2^s}, \dots$ until it becomes smaller than $\frac{1}{dn}$, which immediately gives the desired bound.
\end{itemize}

\subsubsection{Super-Geometrically Increasing Weights}
As the next special case, we will assume that job weights at least double over time -- the only change we will make to handle the general case will be grouping jobs so that groups have exponentially increasing weights. For notational convenience, we assume that exactly one job arrives at each integer time starting from time 0 and index jobs by their arriving time. Our simplified instance is formally defined as follows. Job $t$ arrives at integer time $t \geq 0$ with the following weight $w_t$: $w_0 = 1$ and $W_t := \sum_{t' = 0}^t w_{t'}$ is a power of 2 and is strictly increasing in $t$; thus, we have $W_t \geq 2^t$. 
Note that $i$'s fair share at time $t$ is $\frac{w_i}{W_t}$.
We now show an algorithm that approximately simulates jobs' fair shares. For more intuitive understanding, we advise the reader to read the following pretending that $w_0 = 1$, and $w_t = 2^{t-1}$ for all $t \geq 1$.

\smallskip
\noindent
{\bf Algorithm Description:}  Recursively define a function $g$ defined over positive integers as follows: $g(1) = 1$, $g(2)= 2$, $g(3) = 2^2$, $g(4) = 2^{2^2}$,  $g(k) = 2^{g(k-1)} / 2^{k-1}$ for all integers $k \geq 5$. We extend $g$'s domain to any real number no smaller than $1$ by interpolating the $g$'s values over integer points by arbitrary increasing functions.

In our algorithm a job $i$'s allocation at time $t$ is
$$A(i, t) := \frac{1}{12}\left(g \left( \left \lfloor  g^{-1} \left ( \frac{W_t}{w_i}\right )  \right \rfloor \right)\right)^{-1}$$ 
if $\frac{w_i}{W_t} \ge \frac{1}{12 \cdot 2^i}$; otherwise $A(i, t) = \frac{1}{12 \cdot 2^i}$.

\smallskip
First, we show that each job receives a rate that is (1/12)-approximate.

\begin{lemma} \label{claim:approxWRR}
The algorithm is $\frac{1}{12}$-approximate.
\end{lemma}
\begin{proof}
The claim immediately follows from the fact that $g$ (or equivalently $g^{-1}$) is non-decreasing. Thus, we have $A(i, t) \geq \frac{1}{12} \cdot \frac{w_i}{W_t}$, as desired. 
\end{proof}

The next goal is to show that each job is only disrupted $O(\log^*n)$ times. This observation easily follows if $w_t = 2^{t-1}$ and $g^{-1}$ were $\log^*$ since the value of $\lfloor  g^{-1} ( W_t / w_i )   \rfloor$ would change only very occasionally and $\log^* W_T = \Theta(\log^* n)$; here, $T$ is the last time when a job arrives. But when job weights increase much faster over time, we need more careful analysis. We also need to establish the asymptotic equivalence between $\log^*$ and $g^{-1}$.

\begin{lemma} \label{claim:reallocWRR}
 Each job is disrupted at most $O(g^{-1} (2^{2^{n}}))$ times.
\end{lemma}
\begin{proof}
We group jobs so that all jobs in the same group do not change their weight drastically. Precisely, two jobs arriving at times $t$ and $t+1$ are placed into the same group if 
$W_{t+1} \leq 12 W_t^2$. Let $I_1$, $I_2, \ldots$ be the resulting groups -- jobs in $I_i$ arrive before jobs in $I_{i+1}$. 

For the sake of analysis, fix $i$. We first show that the allocation of every job $j$ in $I_i$ remains unchanged after a job in the next group $I_{i+1}$ arrives. Indeed, at time $t$ when the first job in $I_{i+1}$ arrives, $j$'s fair share, $\frac{w_j}{W_t} \leq \frac{W_j}{W_t} < \frac{W_j}{12 W_j^2} = \frac{1}{12 W_j} \leq \frac{1}{12 \cdot 2^j}$; thus, by the definition of the algorithm, $A(j, t) = \frac{1}{12 \cdot 2^j}$ and $j$'s allocation remains unchanged throughout as $j$'s fair share can only decrease as time progresses. 

  Therefore, we now know that a job $j$ in group $I_i$ can change its allocation only until the last job in $I_i$ arrives. Let $n'$ be the number of jobs in $I_i$. Our goal is to upper bound the number of disruptions of $j$'s allocation by $O(g^{-1}(n'))$. Say the first job in $I_i$ arrives at time $t_1$ and the last job in $I_i$ arrives at time $t_2$. Consider any fixed job $j$ in $I_i$. 
  Observe that $A(j, t') \neq A(j, t)$ only if $\left \lfloor g^{-1} \left ( W_{t'} / w_j \right )  \right \rfloor \neq \left \lfloor  g^{-1} \left ( W_{t} / w_j \right )  \right \rfloor$. Since $g^{-1}$ is increasing,
  it follows that $j$ gets disrupted at most $g^{-1}(W_{t_2}/ w_j) - g^{-1}(W_{t_1} / w_j) + 1$ times. This number is again upper bounded by 
  $O(g^{-1} (W_{t_2} / W_{t_1}))$. This is because $g^{-1}(u v) \leq O(g^{-1}(u) + g^{-1}(v))$, which can be easily seen as $g^{-1}$ is much more slowly growing than $\log$ asymptotically. Thus, we know $j$'s allocation changes at most $O(g^{-1} (W_{t_2} / W_{t_1}))$ times. 
  
  To complete the proof, we only need to upper bound $W_{t_2} / W_{t_1}$   in terms of $n'$. Recall that for any two jobs in $I_i$ arriving at adjacent times $t$ and $t+1$, we have $W_{t+1} \leq 12 W_{t}^2$. Thus, we have $n' = \Omega(\log \log (W_{t_2} / W_{t_1}))$, meaning $W_{t_2} / W_{t_1} = 2^{2^{O(n')}}$. Together with the above observation, we conclude each job $j$ in $I_i$ is disrupted at most $O(g^{-1} (2^{2^{n'}})) = O(g^{-1} (2^{2^{n}}))$, as desired. 
\end{proof}

We now establish the asymptotic equivalence between $\log^*$ and $g^{-1}$.

\begin{lemma}
	For any $n \geq 1$, $\log^* n = \Theta(g^{-1}(n))$. 
\end{lemma}
\begin{proof}
	Note that $\log^{*-1}(2) = g(1)$, $\log^{*-1}(3) = g(2)$, $\log^{*-1}(4) = g(3)$, $\log^{*-1}(5) = g(4)$, and $\log^{*-1}(k+1) \geq g(k)$ for all integers $k \geq 5$. Thus, we have $\log^* n = O(g^{-1}(n))$.  We can also show $g(2k) \geq \log^{*-1}(k)$ by a simple induction on $k$, thus the proof is omitted.
\end{proof}

\begin{corollary} 
 Each job is disrupted at most $O(\log^* n)$ times.
\end{corollary}

To complete the analysis of our algorithm for the simplified instance, it only remains to show that  resource is never over-allocated.

\begin{claim}\label{claim:allocatedWRR}
 	For any integer $k \geq 1$, the value $\frac{1}{12 g(k)}$ appears at most $2\log_2 g(k+1)$ times in the sequence of $f_{t}(t)$, $f_{t-1}(t)$, $f_{t-2}(t)$, $\ldots$, $f_{1}(t)$. 
\end{claim}
\begin{proof}
    Observe that $A(i, t) = \frac{1}{12 g(k)}$ if and only if $g(k) \leq W_t / w_i < g(k+1)$. From the fact that $w_{t+2} / w_t \geq 2$ for all $t \geq 0$, we know that 
the value $\frac{1}{12 g(k)}$ can appear in the sequence at most $2 \log_2 (g(k+1) / g(k)) \leq 2 \log_2 g(k+1)$ times.
\end{proof}

\begin{lemma}
	At any point in time, the total allocation made by the algorithm is bounded by $1$. 
\end{lemma}
\begin{proof}
	For the sake of analysis, we separately handle jobs $j$ with $f_j(t) \leq \frac{1}{12 \cdot 2^j}$ and the other jobs. For the first type of jobs, the total allocation is at most $\sum_{j \geq 0} \frac{1}{12 \cdot 2^j} \leq 1/6$. For the other jobs, we use 	Claim~\ref{claim:allocatedWRR}, which ensure that there are at most  $2\log_2 g(k+1)$ jobs with allocation $\frac{1}{12 g(k)}$. For any $k \geq 1$, we have $2\log_2 g(k+1) \leq 2 g(k)$. In particular, for any $k \geq 4$, we have $2\log_2 g(k+1) \leq 2 g(k) / 2^{k}$. Therefore, 
	 the total allocation for the second type of jobs is at most $\sum_{k \geq 1} 2\log_2 g(k+1) = 2 \frac{g(1)}{12 g(1)} + 2 \frac{g(2)}{12 g(2)} + 2 \frac{g(3)}{12 g(3)} +   \sum_{k \geq 4}  2 g(k) \frac{1}{12 \cdot 2^k g(k)} \leq 2/3$. 
\end{proof}

\subsubsection{Arbitrary Weights}
 Finally, to extend our algorithm to handle the general case, we propose the following pre-processing step that reduces an arbitrary instance to a simplified instance. Conceptually, partition the jobs in the following way. Assume w.l.o.g. that the first job has weight $1$ by scaling.  Intuitively, we would like $G_i$ to consist of the earliest arriving jobs, that are not in $G_0, \ldots, G_{i-1}$,  with aggregate weight $2^i$ (or a higher value that is a power of 2).  In this way, each group will essentially act like a single job of weight $2^i$.  

Formally, our grouping is defined as follows. In our grouping, a job may belong to either exactly one group or two consecutive groups. When job $j$ belongs to only group $G_i$, $j$ remains to have exactly the same weight $w_j$ in the group $G_i$. If $j$ belongs to two groups $G_i$ and $G_{i+1}$, then the sum of $j$'s weight in both groups is exactly its original weight $w_j$. Our goal is to create groups starting from the first job, so that the weight of groups simulate a simplified instance: that is, $w(G_0) = 1$, and $W(G_i)$ is always a power of two, and $W(G_i) \geq  2 W(G_{i-1})$ for all $i \geq 1$, where $W(G_i) := w(G_0) + w(G_1) + \ldots + w(G_i)$. Here $w(G_i)$ denotes the total weight of jobs in $G_i$.

Towards this end, we let  the first group $G_0$ only have the first job. Then, we clearly have $w(G_0) = W(G_0) = 1$. If the second job's weight $w_1$ is  at most 1, since $W(G_0) + w_1 \leq 2 W(G_0)$, the job only belongs to group $G_1$. Otherwise, let $k := \lfloor \log_2 (W(G_0) + w_1) \rfloor$. Then, $W(G_1) := 2^k$, and job $1$ has weight $W(G_0) + w_1 - 2^k$ in group $G_2$ and the remaining weight $w_1 - (W(G_0) + w_1  - 2^k) =  2^k - W(G_0)$ in group $G_1$. In general, suppose job $j$ is to appear in group $G_{i+1}$ (and possibly group $G_{i+2}$). Let $W(G_{i+1})$ be the total weight of jobs that arrived before job $j$. Then, if $W(G_{i+1}) + w_j \leq 2W(G_i)$, then $j$ only belongs to group $G_{i+1}$. Otherwise, $j$ has weight $W(G_{i+1}) + w_j - 2^k$ in $G_{i+2}$ and the remaining weight $2^k - W(G_{i+1})$ in $G_{i+1}$, where $k := \lfloor \log_2 (W(G_i) + w_j) \rfloor$. 

So, when a new job arrives, the above reduction updates groups. If $G_0, G_1, \ldots G_k$ are the groups created, our algorithm pretends that each group is a single job and allocates resources to groups. Note that the last group $G_k$'s total weight may not be a power of two, but then the algorithm pretends that $w(G_k) = W(G_{k-1})$ by a creating a fictitious job of an appropriate weight in $G_k$. Then, the amount of resource allocated to each group is reallocated to individual jobs belonging to the group in proportion to their weight. If a job appears in two groups, then we simply add up the amount of resource to reallocated to the job in both groups.

\smallskip
It now remains to argue why this reduction works. First, we observe that the number of jobs increased by a factor of at most two in the reduction. Therefore, the number of disruptions that occur to each group is still bounded by $O(\log^* n)$. Further, jobs in the same group are disrupted exactly at the same time. Since each job appears in at most two groups, we have an easy conclusion that each job gets disrupted at most $O(\log^* n)$ times. Next, we can see that resource is not over allocated since each job's weight is preserved in the reduction (if it appears in two groups, its weight in both groups is equal to its original weight). Finally, if the current last group $G_k$'s weight is a power of two, it is easy to see that every job gets an $1/12$-approximate of its fair share, as was the case for the simplified instance. Otherwise, the algorithm could pretend more competition by assuming that $w(G_k) = W(G_{k-1})$. However, it could over-estimate the total weight by a factor of at most 2 since $w(G_k) \leq W(G_{k-1})$. Thus, it follows that every job gets an $1/24$-approximate of its fair share. This completes the proof of the upper bound claimed in Theorem~\ref{thm:arrival}.

\subsection{Lower Bound}
This subsection is devoted to proving the lower bound stated in Theorem~\ref{thm:arrival}, bounding the number of disruptions incurred by any $c$-approximate deterministic algorithm $A$ in the arrival model.
The lower bound instance consists of jobs whose weights  geometrically increase.  Job
$i$ has weight $w_i=2^{i-1}$ and arrives at time $i$. 
  Let $W_t$ denote the total weight of the jobs up through job $t$.  
Since the algorithm $A$ is $c$-approximate, it must be the case that for each job $i$ and each time $t \ge i$,  $A(i, t) \ge \frac{w_i}{cW_t}$.    Let $B$ be a matrix where $B_{i,t}$  is $1$ if job $i$'s is disrupted at time $t$ and $0$ otherwise. See Figure~\ref{fig:TriangleIntro}.

\subsubsection{Overview of the Analysis}
Before we formally prove the theorem, we give a high-level overview of the proof. In this overview, certain less important details, such as constant additive terms, will be ignored to make the key idea transparent. The overview will be based on a geometric view of the matrix $B$.  After all, we only need to show that the matrix $B$ has $\Omega(n \log^*n)$ 1s. To count the number of 1s, we will create non-overlapping triangles within the lower triangle matrix. The created triangles are grouped into $K = \Theta(\log^* n)$ groups, $G_1, G_2, \ldots, G_K$. We will let each triangle in $G_k$ have horizontal (or equivalently vertical) length exactly $h(k+1) - h(k)$ for some function $h(k)$ which will be defined shortly. See Figure~\ref{fig:Triangulation}.

We will find $\Omega(n)$ 1s within triangles in each group, which will lead to the desired lower bound $\Omega(n \log^* n)$. Towards this end, we show that for \emph{each} triangle and  jobs $j$ participating in the ``triangle" (the $j$th row from the bottom intersects the triangle), at least half of them must be disrupted. The key idea is to show that if $j$ doesn't get disrupted within the triangle in $G_k$, $j$'s allocation is at least $\frac{1}{2c \cdot 2^{h(k)}}$, where $h(k)$ is some function that we will define. See Figure~\ref{fig:analysis}.

\begin{figure}
  \centering
  \includegraphics[width=.45\textwidth]{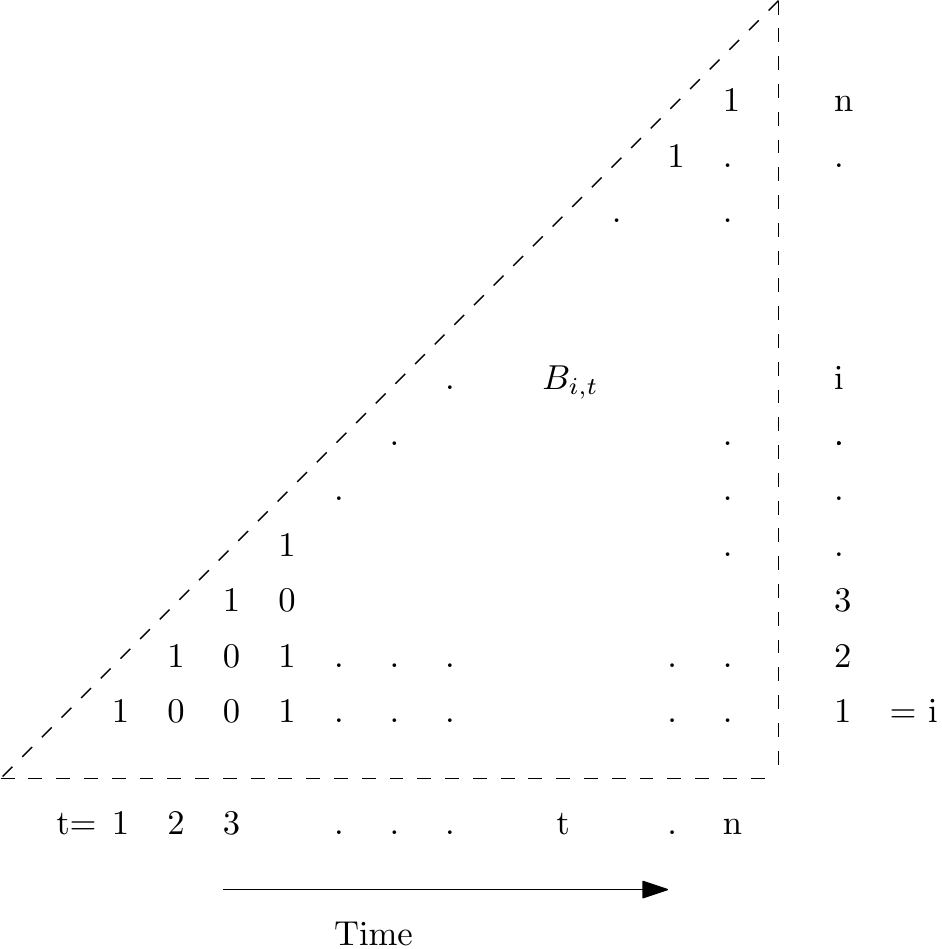}
    \caption{The lower triangle matrix $B$. The entry $B_{i,t}$ is 1 if and only if the job $i$ arriving at time $i$ gets disrupted at time $t$. \label{fig:TriangleIntro}}
\end{figure}

Since there are exactly $h(k+1) - h(k)$ jobs participating in the triangle, if fewer than half of those jobs are disrupted, the amount of resource allocated at the ending time of the triangle is greater than $\frac{1}{2}(h(k+1) - h(k)) \cdot \frac{1}{2c 2^{h(k)}}$. We will define $h$ recursively so that this quantity becomes  more than 1, meaning resource is over allocated.

Thus, we can show that each triangle includes at least half as many 1's as its (either horizontal or vertical) edge length. Due to the disjointness of triangles, this implies that triangles in $G_k$ include 1s whose number is at least half of their edge length, which is $\frac{1}{2}(h(K) - h(k)) \geq h(K) / 4 = n/4$. Since we created $\Theta(\log^* n)$ groups we will have the lower bound.

\subsubsection{Formal Analysis} 
The lower bound instance is formally defined as follows. Job $i$, $i \geq 1$ has weight $w_i=2^{i-1}$ and arrives at time $i$. 
 Job 0 of weight 1 is thought of as a dummy job to keep the cumulative weight to be a power of two. So, we have $W_i := w_0 + w_1 + \ldots w_i = 2^i$.
For a job $i$ and a time $t$ such that $i \leq t$, let  $L(i,t)$ be the latest time $t' \leq t$ such that $B_{i,t}$ is $1$.  In other words, $L(i,t)$ is the last time job $i$'s  was reallocated before time $t$.

We begin with the following claim. 

\begin{claim}\label{claim:lb-rate}
Let $A(i, t)$ denote the allocation to job $i$ at time $t$ by the algorithm $A$.  It must be the case that $A(i, t) \geq \frac{2^{i-1}}{c W_{L(i,t)}} = \frac{2^i}{2c 2^{L(i,t)}}.$
\end{claim}
\begin{proof}
By definition job $i$ was not reallocated during $(L(i,t),t]$. At time $L(i,t)$ job $i$ required a $c$-approximate allocation of its fair share $\frac{2^{i-1}}{2^{L(i,t)}}$ and this remained the same up to time $t$.
\end{proof}

Next,  we formally define disjoint triangles. 
Define the following times recursively.  Let $h(1)  = 1$ and $h(k+1) - h(k) = 8c \cdot 2^{h(k)}$.  Let $h(0) = 0$ for convenience of notation.  Here $k \in [K]$ and $h(K) = n$.  For the lower bound, we will assume that $n$ is chosen such that  $\log^* n$ is integer.  Notice that $K = \Theta(\log^* n)$.  Notice that $\{h(k)\}$ partitions all job arrival times.

Using this, we recursively define \emph{non-overlapping} triangles.  There will be $K$ groups of triangles $G_1, G_2, \ldots G_{K}$.   Let $\ell(k) = h(k+1) - h(k)$.  The length $\ell(k)$ is the length and height of each triangle in $G_{k}$. The group $G_k$ will contain $(n - h(k))/\ell(k)$ triangles.  The $i$th triangle in the group $G_k$ corresponds to jobs $ i\ell(k)+1 $ to $(i+1)\ell(k)$ for $i \in \{0, 1, \ldots (n - h(k))/\ell(k)\}$. Let $J_{k,i}$ be the set of jobs in the $i$th triangle in $G_k$.  Each job in the $i$th triangle will be associated with a set of time steps. The job time pairs will form a triangle with two equal length sides. The $j$th job in the group in arrival order is associated with times $h(k) + i\cdot\ell(k) +(j-1)$ to $h(k) + (i+1)\ell(k)$ for $j = 1, 2, \ldots, \ell(k)$. Let $T_{k,i,j}$ be the time steps corresponding to job $j$ in the $i$th triangle for group $G_k$.

We observe the following.

\begin{claim}
Each job and time pair is associated with at most one triangle. 
\end{claim}

Next we prove a key property of the jobs associated with each triangle.

\begin{lemma}\label{lem:kjt}
Consider any group $G_k$ and the $i$th triangle in the group. Let  $t = h(k) + (i+1)\ell(k)$ be the last time associated with the triangle. For each job $j \in J_{k,i}$ if $j$ is not disrupted during a time in $T_{k,i,j}$ then  $L(j,t)-j \leq h(k)$.
\end{lemma}
\begin{proof}
Fix a group $k$. We prove the lemma by induction on the jobs $j$. Formally, we show a stronger statement where $t_j-j \leq h(k-1)$  where $t_j$ is the earliest time in $T_{k,i,j}$ and $i$ is the unique triangle $j$ contributes to in the $k$th group. First consider the case where $i=0$ and fix $j=1$, the lowest indexed job  in $J_{k,i}$.  By definition $T_{k,i,j}$ contains all of the times from $h(k)$ to $h(k)+ \ell(k) = h(k+1)-1 $.   Thus, if $j$ is not disrupted during $T_{k,j,i}$ then the last time $j$ was reallocated was before time $h(k)$, and $L(j,t) \leq h(k)$.  

Now consider any $j$ and it's associated triangle $i$ in the $k$th group.  Inductively, we know that $t_{j-1}-(j-1) \leq h(k)$ for job $j-1$.  By definition of the triangles $t_j = t_{j-1}+1$.  Thus $t_j$ increases by one, implying  $t_{j}-j = t_{j-1}-(j-1) \leq h(k)$ and the lemma follows.
\end{proof}

The next lemma bounds the number of disruptions for job and time pairs inside each triangle.

\begin{lemma}\label{lem:triallocate}
Fix a group $G_k$ for $k \in [K-1]$ and the $i$th triangle in the group for  $i \in \{0, 1, \ldots (n - h(k))/\ell(k)\}$. It is the case that at least $\ell(k)/2$ jobs $j$ in $J_{k,i}$ are disrupted at their corresponding time steps in  $T_{i,j,k}$.
\end{lemma}

\begin{proof}
For the sake of contradiction say that it is not the case.  Let $t = h(k) + (i+1)\ell(k)$ be the latest time in $T_{i,k,j}$ for all jobs $j$ associated with the triangle.  By Claim~\ref{claim:lb-rate}, job $j$ must be processed at a rate of $\frac{2^{j-1}}{c2^{L(j,t)}} = \frac{1}{2c2^{L(j,t)-j}} \geq \frac{1}{2c 2^{h(k)}}$ at time $t$. The last inequality follows from Lemma~\ref{lem:kjt}. There are $\ell(k)$ jobs in $J_{k,i}$.  If less than $\ell(k)/2$ are reallocated during their corresponding times in $T_{i,j,k}$, then their total allocation at time $t$ is greater than the following.

\begin{eqnarray*}
&& \frac{1}{2}\ell(k) \cdot  \frac{1}{2c 2^{h(k)}} \\
&=&\frac{1}{2}(h(k+1) - h(k)) \cdot  \frac{1}{2c 2^{h(k)}} \;\;\;\; \mbox{[Definition of $\ell(k)$]}\\
&\geq& 2 \;\;\;\;\mbox{[Definition of $h(k+1) = h(k) + 8c2^{h(k)}$]}
\end{eqnarray*}

This contradicts the total available amount of resource being $1$ at time $t$, and the lemma follows.
\end{proof}

We are now ready to prove the lower bound in Theorem~\ref{thm:arrival}.
Consider any group of triangles $G_k$ for $k \in [K-1]$. The jobs indexed $1$ to $((n - h(k))/\ell(k) +1)\ell(k)$ are associated with the group.  Notice that $h(k) \leq \frac{n}{2}$ for $k \in  [K-1]$ by definition of $h(k)$.  Thus, the number of jobs within each group is at least $n/2$. By definition each job appears in at most one triangle in the group.  Further, by Lemma~\ref{lem:triallocate} half of the jobs associated with each triangle are disrupted at times associated with the triangle.  Thus, there are at least $\frac{n}{4}$ disruptions at time job pairs associated with triangles in group $G_k$.    Knowing that there are $\Theta(\log^*n)$ groups $G_k$, we have that the total number of disruptions is $\Omega(n\log^*n)$, proving the lower bound.

\begin{figure}[htbp]
  \centering
  \includegraphics[width=.45\textwidth]{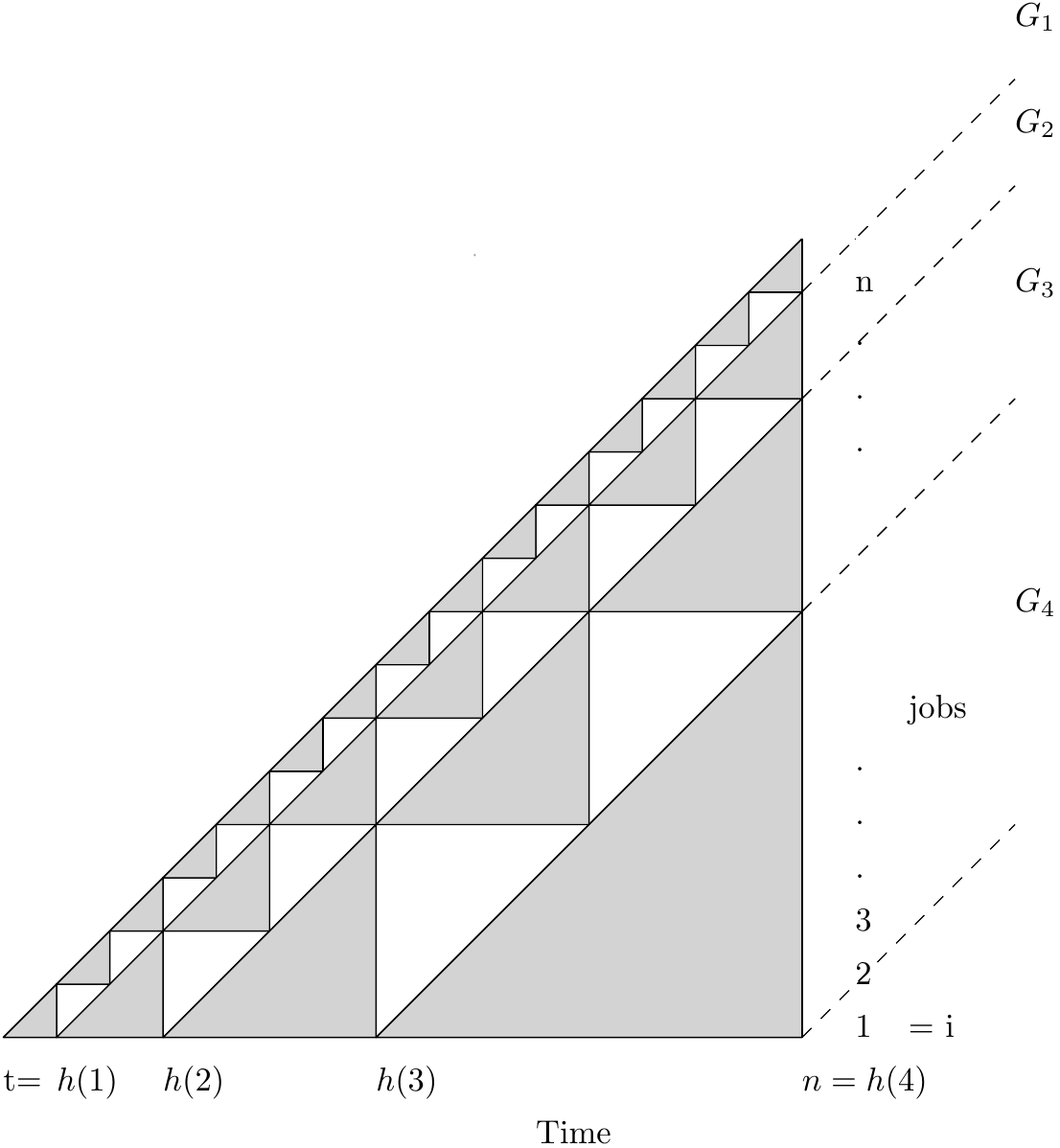}
    \caption{Disjoint triangles within the lower triangle matrix $A$. Here, each triangle includes no entries on the diagonal edge. Triangles are grouped into $G_1, G_2, G_3, \ldots, G_K$. All triangles in $G_k$ has length $h(k+1) - h(k)$.      For better visualization, this figure assumes that $h(k) = 2^{k}-1$ but the actual function $h(k)$ we use in the formal proof is  very similar to $\log^{*-1} k$, which grows much faster. Also, this figure assumes $K = 4$. 
    \label{fig:Triangulation}}
\end{figure}

\begin{figure}[htbp]
  \centering
  \includegraphics[width=.45\textwidth]{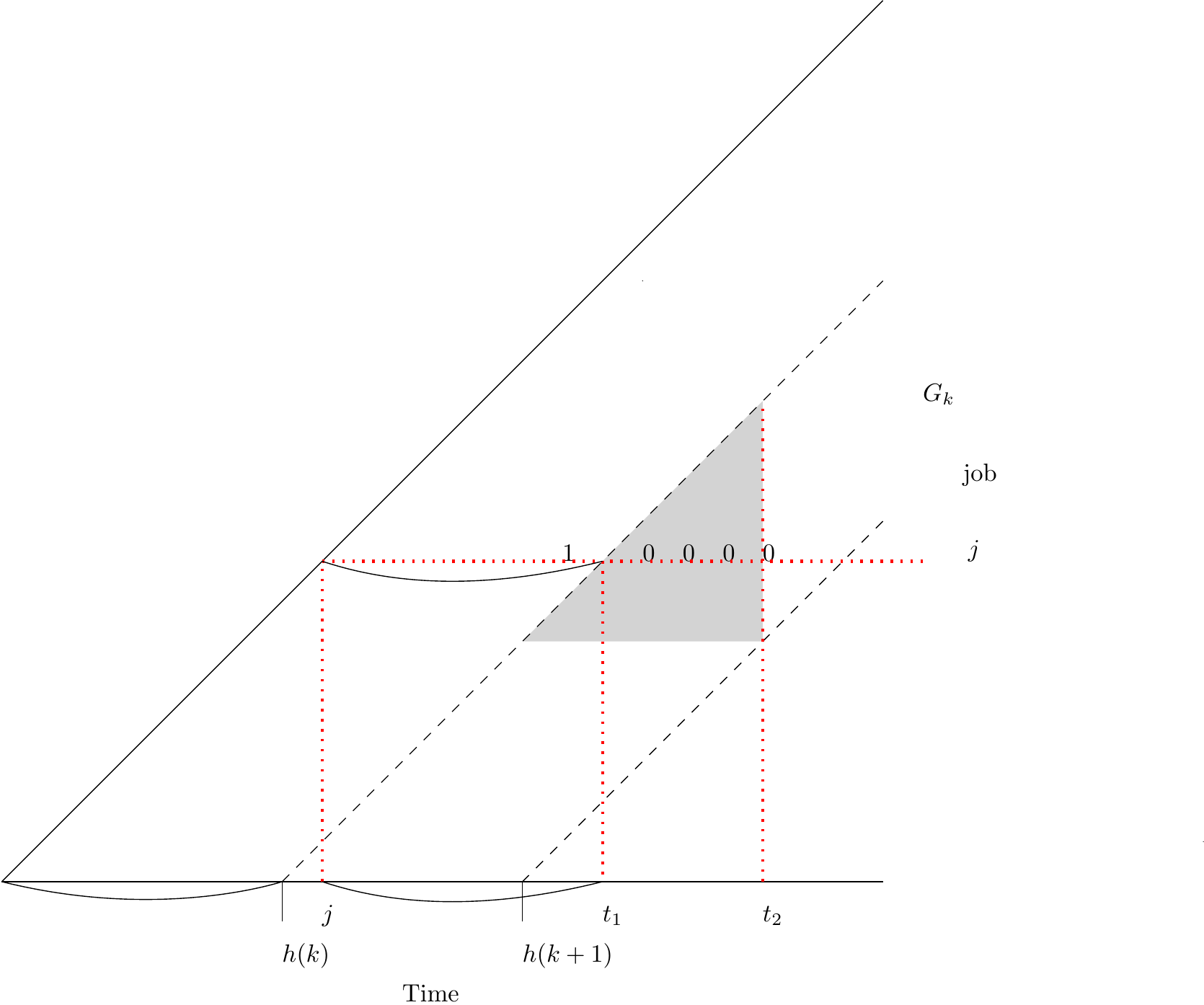}
    \caption{In this figure, a triangle $T$ in group $G_k$ is shaded. Note that the triangle has length exactly $h(k+1) - h(k)$. The job $j$ participates in the triangle, as the horizontal line $l$ includes $j$ intersects the triangle $T$. This figure illustrates how we can show that $j$ has a significantly large allocation at time $t_2$, which is the ending time of $T$, if $j$ is not disrupted within $T$. This  is why the line $l$ has no 1s within $T$ in the figure. Since $j$ is not disrupted during $(t_1, t_2]$, $j$'s allocation must be at least $\frac{w_j}{c \cdot W_{t_1}} = \frac{1}{2c 2^{t_1 - j}} = \frac{1}{2c 2^{h(k)}}$. The fact $t_1 - j = h(k)$ is immediate from the three edges indicated by arcs. 
    \label{fig:analysis}}
\end{figure}

\section{Weighted Fairness: Departures and Proof of Theorem~\ref{thm:departure}}
\label{sect:arrivaldeparture}

This section is devoted to proving     Theorem~\ref{thm:departure}. 
We first discuss intuitions on why minimizing the number of disruptions is more challenging if jobs can both arrive and depart. Recall that in the arrival model a job's allocation has to change if its fair share significantly changes -- say  from $1/s$ to $1/ 2^s$. In other words, this requires the arrival of many new jobs of higher weights. The next reallocation requires even many more new jobs' arrival to further increase the total weight of jobs. However, if jobs can depart we can repeat this process without keeping increasing the total weight. Thus, we can effectively create an instance where a large number of jobs are repeatedly disrupted. 

Before we prove the theorem, we reproduce it for convenience.

\begin{startheorem}
Consider weighted fair share policies with both job arrivals and departures.
For every $c$-approximate deterministic algorithm $A$, there is an instance that causes $A$ to make $\Omega(  n^{1+ 1/(4c+1)})$ disruptions. 
\end{startheorem}

Time will be divided into phases. At each time $t$ in a phase, other than the last time, a batch of jobs arrive, and no jobs depart. Batches will be of different types. 
A type
$k$ batch consists of $M/b^k$ jobs, each with weight $a^k$. 
Here $a = 2b \ge 4$. There will be a most one batch of each type alive at any time. So when the time is understood, we will use $B_k$ to refer to the batch of type $k$ alive at that time. So $B_k$ has a factor $b$ more jobs than $B_{k+1}$, but the jobs in $B_{k+1}$ have aggregate weight that is twice the aggregate weight of the jobs in $B_k$. 
We say a batch $B_k$ is disrupted at a particular time if at least half of the jobs in the batch have been disrupted since the last time
that the largest  batch type of an alive batch was $k$. 

At time 0, a type 0 batch $B_0$ arrives. Now consider a time $t > 0$. 
If the alive batches are $B_0, \ldots, B_{k-1}$, and no batches have yet been disrupted in this phase, then a type $k$ batch $B_k$  arrives at time $t$. If on the other hand, there was a batch that was disrupted at time $t-1$,  let $k$ be the smallest type such that $B_k$ was disrupted at time $t-1$. Then time $t$ is the last time in this phase, and jobs in batches $B_j$, $j > k$, depart at time $t$. Thus heading into the next phase,
the alive batches are $B_0, \ldots, B_k$.
 The input terminates after the first phase where at least $3M$ jobs have arrived over all phases.

We begin the analysis by bounding $A$'s allocation of the resource
to a batch $B_k$ between a time $r$  such that $B_k$ was the alive batch of highest
type,  until the next time $s$ that $B_k$ is disrupted
We claim that in aggregate the jobs in $B_k$ must
be allocated a $\frac{1}{4c}$ fraction of the resource at each time
in the range $[r, s]$.
To see why this is the case, consider the time $r$.
The aggregate weight of the alive jobs in batches $B_0, \ldots B_{k-1}$ is
$\sum_{i=0}^{k-1} \frac{M a^i}{b^i} \le \left( \frac{M a^{k}}{b^k}\right)\left(\frac{b}{a -b}\right) = \frac{M a^{k}}{b^k}$.
Thus the aggregate weight the jobs in $B_k$ is at least  the aggregate weight of the jobs in batches $B_1, \ldots, B_{k-1}$.
Thus as $A$ is $c$-approximate,
in aggregate the jobs in $B_k$  must be allocated a $\frac{1}{2c}$ fraction of the resource at time $r$. Thus at time $r$,  each job in $B_k$ must be allocated a $\frac{b^{k}}{2cM}$ fraction
of the resource. 
As long as $B_k$  is not disrupted, at least half the jobs in $B_k$ must thus be allocated 
a $\frac{b^{k}}{2cM}$ fraction of the resource. Thus we conclude
in aggregate the jobs in the batch must be allocated a $\frac{1}{4c}$ fraction of the resource until $B_k$ is disrupted.
Thus  no batch of type $4c+2$ can ever arrive, as a batch of type
$4c+1$ must cause a disruption.

Now consider the time $s$ when $B_k$ was disrupted,
meaning the number of jobs disrupted in $B_k$ since time $r$ 
is at least $\frac{M}{2b^k}$. Let the batches
of higher type at time $s$ be  $B_{k+1}, \ldots B_\ell$. 
These disruptions in $B_k$ are then charged equally to the 
departing jobs in $B_{k+1}, \ldots, B_\ell$. As the number of jobs in 
$B_{k+1}, \ldots B_\ell$ is 
$\sum_{j=k+1}^\ell M/b^j \le \sum_{j=k+1}^\infty \frac{M}{b^j} = \left( \frac{M}{b^k} \right) \left( \frac{1}{b-1} \right)$, each job in $B_{k+1}, \ldots B_\ell$ is charged at least $\frac{b-1}{2}$. 
As $\sum_{j=0}^\infty \frac{M}{b^j} \le 2M$ when $b \ge 2$, there are at most $2M$ jobs at the end that have not been charged. Thus at least $M$ jobs have been charged at least $\frac{b-1}{2}$. Thus the number of disruptions per arrival is at least $\frac{b-1}{6}$.

In order for this construction to be well defined, we need that that highest type batch has
at least 1 job. So we need that $\frac{M}{b^{4c+1}} \ge 1$, as $ 4c+1$
is the highest possible batch type. 
This is equivalent to $b \le M^{1/(4c+1)}$.
Thus we can conclude that the number of disruptions for
$A$  is $\Omega(  n^{1+ 1/(4c+1)})$. And as $W \le a^{4c+1} = (2b)^{4c+1}$, the number
of disruptions caused by $A$ is also $\Omega(n W^{1/(4c+1)})$.

\section{Weighted Fairness: Random Weights and Proof of Theorem~\ref{thm:randomdeparture}}
\label{sect:randomarrivaldeparture}

This section is devoted to proving Theorem~\ref{thm:randomdeparture}, which
we reproduce here for convenience.

\begin{startheorem}
Consider weighted fair share policies with both job arrivals and departures.
Assume that the weights $w_1, \ldots, w_n$ of the jobs are arbitrary,
but the assignment of these weights to the $n$ jobs is uniformly random.\footnote{This random permutation model clearly includes the case when jobs' weights are sampled i.i.d.}  
In this setting there is an  $4$-approximate randomized algorithm $A$,
for which the expected number of disruptions per arrival and per departure is 
at most 5. 
\end{startheorem}

Let us for convenience assume that the smallest weight is 1. 
Consider the following algorithm $A$:

\medskip
\noindent
{\bf Description of Algorithm A:} Initially let $T$ be a  random number 
in the range $[1/2, 1]$. Define a threshold to be any weight of the form $2^k T$ for
some integer $k$. 
When a  job $j$ arrives, its allocation is set to half of its weighted fair share. If this arrival causes the total weight in the
system to cross a threshold (that is to increase from below a threshold to above a threshold)
then the allocation of every job is reset to half of its weighted fair share.
Similarly, if the departure of a job $j$ causes the total weight in the
system to cross a threshold (that is to decrease from 
above a threshold to below a threshold)
then the allocation of every job is reset to half of its weighted fair.

Because the initial allocation
for a job is half of its fair share, and a reset must happen by the time that the total weight doubles, the resource will not be over allocated. 
Because the initial allocation
for a job is half of its fair share, and a reset happens must happen
by the time that total weight halves,
the algorithm $A$ maintains 4-approximate fairness.

Thus we are left to bound the expected number of disruptions per arrival and departure. We will only give the analysis for arrivals, as the analysis for 
departures more or less follows by symmetry. 
Assume that there are $k-1$ jobs in the system when a new job arrives.
For convenience let us renumber the earlier arriving jobs to  $1$ to $k-1$
and the new job to $k$. Let us condition on the weights of these $k$
jobs being $w_1 \le w_2  \ldots \le  w_k$. 
Let $W = \sum_{j=1}^k w_j$. 
Let $w$ be the random variable denoting
the weight of job $k$. (Note that due to our random assignment assumption,
$w$ is not necessarily $w_k$.) Let $D$ be the total number of disruptions 
caused by $k$'s arrival, and let $E_k$ be the event that $k$'s arrival caused
a reset.
Then 
\begin{align*}
    E[D]  = & \;E[D \mid w=w_k] \cdot P[w=w_k] + \\
    &\sum_{j=1}^{k-1} E[D \mid w_k=k \mbox{ and } E_k] \cdot P[E_k \mid w = w_j] \cdot P[w=w_j]  \\
    \le & \; k \cdot (1/k) + \sum_{j=1}^{k-1}  k \cdot  (4 w_j/W) \cdot (1/k) \\
    =& \; 1 + \sum_{j=1}^{k-1}   (4 w_j/W) \\
    \le & \; 5
\end{align*}
To elaborate, for all $j \in [k]$, $P[w=w_j] = 1/k$.
If $w=w_k$ then it is possible that $k$'s arrival caused
the total weight to cross a threshold,
and and thus $A$ would disrupt all $k$ of the
jobs in the system. Still $E[D \mid w_k=k] \le k$.
If $w= w_j$, $j < k$, then we know that
$k$'s arrival does not more than double the weight. So before
job $k$'s arrival, the total weight was at least $W/2$. 
As there can be at most two threshold between  $W/2$ and $W$,
$P[E_k \mid w = w_j] \le (4 w_j/W)$. 
And again obviously $E[D \mid w_k=k \mbox{ and } E_k] \le k$.

Note that Theorem \ref{thm:randomdeparture} would still hold if job
weights were drawn i.i.d. from some distribution.

\section{Monotone Fairness: Proof of Theorem~\ref{thm:monotone}}
\label{sect:monotone}

This section is devoted to proving Theorem~\ref{thm:monotone}, which we 
reproduced here for convenience.

\begin{startheorem}
Consider general monotone share policies in the arrival-only model.
There is a $(1+\epsilon)$-approximate deterministic algorithm $A$ such that the number of disruptions per job is $O(\log  n)$. This bound is tight, that is, for every deterministic $O(1)$-approximate algorithm $A$, there are instances that cause $A$ to make $\Omega(n \log n)$ disruptions.
\end{startheorem}

We will first design and analyze an algorithm $A$ to prove the upper bound
portion of the theorem.
As jobs arrive over time, we will maintain the invariant that for every job $j$, if $A$ has allocated $1/c^i$ of the resource to $j$, 
then the fair share for that job will be $1/c^{i+1}$. When a new job $j$ arrives, its fair share is set to $1/c$. 
Because of the invariant, and the feasibility of $A$, the aggregate un-apportioned fair share  before $j$ arrives is at least $1-1/c$. 
So the fair share is not exceeded when job $j$ arrives. By the fairness of $A$, it must allocate at least $1/c^2$ fraction of the
resource to job $j$. Note that for any job $i$ disrupted when job $j$ arrives, $A$'s allocation to $i$ can decrease by at most a factor of $c^2$,
or this would contradict the fairness of $A$. 

After the algorithm $A$ decides how much to allocate to each of the jobs after $j$'s arrival then the adversary updates the fair share allocations of the jobs. The fair share of job $j$, and any jobs that $A$ disrupted in response to $j$'s arrival, 
is set to $1/c$ of whatever $A$'s allocation is at this time. The process then proceeds with the arrival of job $j+1$.

Now note that after $n$ jobs arrive, at least $n/c$ jobs have fair share no more than $c/n$.  This follows by an averaging argument and the fact that at most a unit of fair share is allocated in aggregate to the jobs. Since these $n/c$ jobs started with allocation at least
$1/c^2$, and decreased by a factor of at most $c^2$ per disruption, the number of disruptions $d$ for each of these
jobs satisfies $(\frac{1}{c^2}) (\frac{1}{c^{2d}}) \le \frac{c}{n}$, or equivalently $d \geq (\log_c n  - 3)/2$.  The theorem follows.

We now turn to proving a lower bound for an arbitrary algorithm $A$.
Assume for the moment that $n$ is known a priori. 
A job is called light if its fair share is at most 
$\frac{\epsilon}{2n}$, and is called heavy otherwise. 
When a new job arrives, $A$ allocates that job its fair share divided by $1 + \epsilon$. If a job $j$'s fair share is lowered, there are three
possible responses. If $j$ was already light before its fair share was
lowered, then $A$ doesn't change its allocation. 
Otherwise, if the new fair share is still greater than a $1+ \epsilon/2$ factor of $A$'s allocation, then again $A$ doesn't change its allocation.
Finally, 
if the fair share of heavy job is ever lowered to be less than a $1+ \epsilon/2$ factor of $A$'s current allocation,
then $A$ resets its allocation for this job to its be its new fair share divided by $1 + \epsilon$. 

The total allocation of $A$ on light jobs is at most $\frac{\epsilon}{2}$.
At any particular time, let $F$ be the total fair share of the heavy jobs, and $U$ the unallocated fair share. 
Thus the total unallocated portion of the resource for $A$ is always at least $\frac{\epsilon F}{2} - \frac{\epsilon}{2} + U \ge 0$. 
Thus $A$ never overuses the resource. Also its obvious that $A$ has $O(\log n)$ disruptions per job. If $n$ is not known a priori, then the standard guess-and-double technique can be used; So when $A$'s estimation of $n$ doubles, it doubles its estimation, and resets all allocations to be what they would have been with the new estimate.

\section{Conclusions}
We have presented simple policies for minimizing the number of reallocations needed to maintain fair resource shares under both arrival-only and arrival-departure models. We have shown that in the worst case, no better policies exist. We have also presented a stochastic model where the assignment of weights is independent of arrival and departure times, and have shown improved results in this model. 

We conclude with some open questions. First, though we have a result for general monotone allocations with $D$ resources, it would be interesting to specialize it to specific classes of fair share policies. More specifically, we conjecture that the number of disruptions needed for  DRF  is significantly better than our bound of $O(n D \log n)$ in the arrival only model. Secondly, our model assumes the departure time of a job is independent of the rate allocated to it. In reality, the job has a fixed amount of processing, so the departure time will depend on the rate allocated to the job. Can we modify our model to handle this aspect? Finally, it would be interesting to study the computational complexity of the {\em offline} setting, where all arrivals and departure times are known in advance, and the goal is to compute the instance-optimal tradeoff between disruptions and fairness.

\paragraph{Acknowledgment.} Im is supported in part by NSF grants CCF-1409130, CCF-1617653, and CCF-1844939. Moseley is supported in part by a Google Research Award and NSF grants CCF-1617724, CCF-1733873 and CCF-1725543. Munagala is supported in part by NSF grants CCF-1408784, CCF-1637397, and IIS-1447554, ONR award N00014-19-1-2268, and awards from Adobe and Facebook. Pruhs is supported in part by NSF grants CCF-1421508 and CCF-1535755, and an IBM Faculty Award.

\bibliographystyle{plain}
\bibliography{refs,mmf}


\end{document}